\documentclass[11pt]{article}
\usepackage[latin9]{inputenc}
\usepackage{geometry}
\usepackage[american,english]{babel}
\usepackage{float}
\usepackage{amsmath}
\usepackage{amsthm}
\usepackage{amssymb}
\usepackage[authoryear]{natbib}
\usepackage[unicode=true,pdfusetitle,
 bookmarks=true,bookmarksnumbered=false,bookmarksopen=false,
 breaklinks=false,pdfborder={0 0 0},pdfborderstyle={},backref=false,colorlinks=false]
 {hyperref}

\makeatletter
\theoremstyle{definition}
\newtheorem{defn}{\protect\definitionname}
\theoremstyle{plain}
\newtheorem{prop}{\protect\propositionname}
\theoremstyle{remark}
\newtheorem{rem}{\protect\remarkname}
\theoremstyle{definition}
 \newtheorem{example}{\protect\examplename}
\theoremstyle{plain}
\newtheorem{cor}{\protect\corollaryname}
\theoremstyle{plain}
\newtheorem{lem}{\protect\lemmaname}

\bibpunct{(}{)}{,}{a}{,}{,}
\usepackage{url}
\theoremstyle{plain}

\usepackage{lscape}

\usepackage[bottom]{footmisc}

\makeatother

\addto\captionsamerican{\renewcommand{\corollaryname}{Corollary}}
\addto\captionsamerican{\renewcommand{\definitionname}{Definition}}
\addto\captionsamerican{\renewcommand{\examplename}{Example}}
\addto\captionsamerican{\renewcommand{\lemmaname}{Lemma}}
\addto\captionsamerican{\renewcommand{\propositionname}{Proposition}}
\addto\captionsamerican{\renewcommand{\remarkname}{Remark}}
\addto\captionsenglish{\renewcommand{\corollaryname}{Corollary}}
\addto\captionsenglish{\renewcommand{\definitionname}{Definition}}
\addto\captionsenglish{\renewcommand{\examplename}{Example}}
\addto\captionsenglish{\renewcommand{\lemmaname}{Lemma}}
\addto\captionsenglish{\renewcommand{\propositionname}{Proposition}}
\addto\captionsenglish{\renewcommand{\remarkname}{Remark}}
\providecommand{\corollaryname}{Corollary}
\providecommand{\definitionname}{Definition}
\providecommand{\examplename}{Example}
\providecommand{\lemmaname}{Lemma}
\providecommand{\propositionname}{Proposition}
\providecommand{\remarkname}{Remark}

\begin{document}
\title{The Use of Symmetry for Models with Variable-size Variables}
\author{Takeshi Fukasawa\thanks{Waseda Institute for Advanced Study, Waseda University. 1-21-1, Nishiwaseda, Shinjuku, Tokyo, Japan. E-mail: fukasawa3431@gmail.com.\protect \\
This study is supported \foreignlanguage{american}{JSPS KAKENHI Grant Number JP24K22629.}}}
\maketitle
\begin{abstract}
This paper presents a universal representation of symmetric (permutation-invariant) functions with multidimensional variable-size variables. These representations help justify approximation methods that aggregate information from each variable using moments. It further discusses how these findings provide insights into game-theoretic applications, including two-step policy function estimation, Moment-based Markov Equilibrium (MME), and aggregative games.

Regarding policy function estimation, under certain conditions, estimating a common policy function as a function of a firm\textquoteright s own state and the sum of polynomial terms (moments) of competitors\textquoteright{} states is justified, regardless of the number of firms in a market, provided a sufficient number of moments are included. For MME, this study demonstrates that MME is equivalent to Markov Perfect Equilibrium if the number of moments reaches a certain level and regularity conditions are satisfied. 

Regarding aggregative games, the paper establishes that any game satisfying symmetry and continuity conditions in payoff functions can be represented as a multidimensional generalized aggregative game. This extends previous research on generalized (fully) aggregative games by introducing multidimensional aggregates.

\end{abstract}
\pagebreak{}

\section{Introduction}

Many economic problems require approximation methods due to their complexity and computational burden. Examples include two-step policy function estimation methods for dynamic models (\citealp{hotz1993conditional,bajari2007estimating}),\footnote{In two-step estimation methods, policy functions used in the second step are ``approximated'' using those consistently estimated in reduced form in the first step. } and Moment-based Markov Equilibrium (MME; \citealp{ifrach2017framework}) as an approximation of Markov Perfect Equilibrium (MPE). In these methods, accurately approximating functions (e.g., value functions, policy functions) is crucial for drawing valid conclusions.

One challenge in function approximation arises from variable-size variables. Consider a dynamic investment competition model as in \citet{Ericson1995}, where the capital stocks of all firms at the beginning of each period serve as state variables. Policy function corresponds to the density of the value of each firm's investment as a function of states. How, then, should we estimate the policy function using data from multiple markets in a reduced form, given that some markets have three firms while others have only two? Here, the states are variable-size, because the size of states varies across markets.

This study presents universal representations of multidimensional symmetric (permutation invariant) functions with variable-size variables using polynomial functions (Section \ref{sec:Universal-representation}), drawing on prior research in machine learning and mathematics, including \citet{zaheer2017deep}. Symmetry appears in various economic applications. For instance, in models of firms' dynamic oligopolistic investment competition, it is frequently assumed that the order of each competing firm's states does not affect a firm's policy function. Similarly, in a model of product differentiation, the order of pairs of each competing firm's product price and product characteristics often does not influence a firm's profit function in many demand models. This study discusses how these findings provide insights into game-theoretic applications, including two-step policy function estimation, MME, and aggregative games\footnote{The latter two primarily consider settings with fixed-size variables.} (Section \ref{sec:Economic-Applications}). The results presented in the current study help assess the validity of approximation methods that aggregate information from each variable using moments. 

For the policy function estimation example, under certain conditions, estimating a common policy function as a function of a firm\textquoteright s own state and the sum of polynomial terms (moments) of competitors\textquoteright{} states is justified, regardless of the number of firms in a market, as long as a sufficient number of moments are included. Concerning the MME, this study demonstrates that MME is equivalent to the MPE if the number of moments reaches a certain threshold and regularity conditions are satisfied. Regarding aggregative games, we can show that any game can be represented as multidimensional generalized aggregative, introducing multidimensional aggregates into the generalized (fully) aggregative games (\citealp{cornes2012fully}). 

The rest of this paper is organized as follows. Section \ref{sec:Literature} describes the relationship to previous studies. Section \ref{sec:Universal-representation} presents the mathematical results. Section \ref{sec:Economic-Applications} discusses the three game-theoretic applications mentioned above. In this section, we briefly review the significance and then examine how our mathematical results from the previous section provide insights into them. Finally, Section \ref{sec:Conclusions} presents the conclusions.

In Appendix \ref{sec:Additional-results-symmetry}, we further discuss the implications for economic models. Appendix \ref{subsec:Merger} examines merger analyses. When a merger occurs in a market, the number of firms changes, and this can be addressed by the framework of variable-size variables. Appendices \ref{subsec:Models-with-dynamic-demand} and \ref{subsec:Models-with-Multi-product} explore models with dynamic demand and multi-product firms. Appendix \ref{sec:Proof-symmetry} presents all the proofs of the main propositions. 

\section{Literature\label{sec:Literature}}

First, this study relates to and contributes to economic research based on game-theoretic models. It is sometimes not easy to directly estimate or solve the models, and some approximation techniques in a broad sense are introduced (e.g., two-step policy function estimation method, MME). By showing general mathematical results based on the concept of symmetry, this study evaluates the justification of these methods from a different perspective. The mathematical results are broad and straightforward, making them useful for further research.

Note that the idea of using symmetry is not new in the literature. For instance, \citet{Pakes_McGuire1994} discussed using symmetry to reduce the dimension of state variables in a dynamic game with finite state space. However, their approach is applicable only to models with a fixed-size finite state space. In contrast, the mathematical results in the current paper can be also applied to dynamic models with both fixed-size/variable-size and discrete/continuous state space, and they are more general. 

\citet{kahou2021exploiting} also discussed using symmetric structures with fixed-size variables, mainly for quantitatively solving heterogeneous agent macroeconomic models using deep learning techniques.\footnote{\citet{han2025deepham} also developed a method for solving heterogeneous agent models. They utilized a flexible and interpretable representation of the agent distribution through generalized moments, which extracts key information of the distribution. To justify their strategy, they relied on the universal approximation of fixed-size permutation invariant functions demonstrated by \citet{han2022universal}. } The current study complements their analysis by formally showing the representations allowing for the case of variable-size variables. Note that they speculated in Section 6.3 of their paper that exploiting the symmetric structure enables researchers to solve dynamic models with networks. Typically, network structure is characterized by variable-size variables, because the number of each agent's neighbors is heterogeneous.

Finally, this study builds on the machine learning and mathematics literature on symmetric functions. Recent machine learning studies have focused on leveraging symmetric structures for tasks such as regression and classification. \citet{zaheer2017deep} developed a method called ``deep sets'',\footnote{Symmetric functions can be thought of as a function on sets.} utilizing symmetric structure and deep learning techniques, and showed good performance in some machine learning applications. To formally justify the strategy, \citet{zaheer2017deep} showed a universal representation of symmetric functions with single-dimensional fixed-size variables. \citet{chen2024representation} extended the result to the setting with multi-dimensional fixed-size variables. The current study extends the mathematical results to the setting with multi-dimensional variable-size variables, and derives a result more relevant to game-theoretic applications.

As discussed in Remark \ref{rem:wagstaff_2022} of Section \ref{sec:Universal-representation} in detail, \citet{wagstaff2022universal} derived a representation of symmetric functions with single-dimensional variable-size variables. While they demonstrated the existence of a function that aggregates information from variable-size variables, the exact form of the function remains unclear. In contrast, this study employs polynomial functions to aggregate information from variable-size variables, resulting in a simpler and more intuitive representation.

\section{Universal representation of symmetric functions with variable-size variables\label{sec:Universal-representation}}

In the following discussions, let $I,J,J^{\prime}$ be integers, and $\mathcal{S}_{J^{\prime}}$ be the set of all permutations of $J^{\prime}$ elements. Define $\Omega^{J^{\prime}}\subset\mathbb{R}^{I\times J^{\prime}}$ as a subset of $\mathbb{R}^{I\times J^{\prime}}$. Let $V:\Omega^{J^{\prime}}\subset\mathbb{R}^{I\times J^{\prime}}\rightarrow\mathbb{R}$ be a function expressed in the form $V\left(x_{1},\cdots,x_{J^{\prime}}\right)$ where $x_{j}\ \left(j=1,\cdots,J^{\prime}\right)$ is a $I$-dimensional variable. The restriction of \ensuremath{V} to a subset $\Omega^{J^{\prime}}$ is denoted by $V|_{\Omega^{J^{\prime}}}$. Finally, let $\mathcal{J}$ be a set; for instance, $\mathcal{J}$ can be $\{1,\cdots,J^{\prime}\}$. 

In economic applications, $\mathcal{J}$ can can represent the set of firms, products, or agents. $V\left(\left\{ x_{j}\right\} _{j\in\mathcal{J}}\right)$ may describe policy functions or profit functions, where $x_{j}$ denotes agent $j$'s state. The capital stock of firm $j$, or the pair of firm $j$'s product price and quality, can be onsidered its states. 

First, we define permutation invariance, a particular form of symmetry.
\begin{defn}
A function $V:\Omega^{J^{\prime}}\subset\mathbb{R}^{I\times J^{\prime}}\rightarrow\mathbb{R}$ is permutation invariant,\footnote{In some literature in industrial organization (e.g., \citealp{doraszelski2007framework}), the property is called anonymity or exchangeability. They discuss the use of symmetry in dynamic models with discrete states. Our discussion is more general in that we allow for the case with continuous states.} if, for all permutations $\Pi\in\mathcal{S}_{J^{\prime}}$, $V\left(\Pi\left(x_{1},\cdots,x_{J^{\prime}}\right)\right)=V\left(x_{1},\cdots,x_{J^{\prime}}\right)$. A function $V:\cup_{J^{\prime}\leq J}\Omega^{J^{\prime}}\subset\mathbb{R}^{I\times(\leq J)}\rightarrow\mathbb{R}$ is permutation invariant if $V|_{\Omega^{J^{\prime}}}$ is permutation invariant for every $J^{\prime}\leq J$.
\end{defn}
For instance, consider a fixed-size function $f:\mathbb{R}^{3}\rightarrow\mathbb{R}$ such that $f(x_{1},x_{2},x_{3})=x_{1}x_{2}+x_{1}x_{3}+x_{2}x_{3}$. Similarly, consider a variable-size function $g:\left(\mathbb{R}^{2}\cup\mathbb{R}^{3}\right)\rightarrow\mathbb{R}$, defined as $g(x_{1},x_{2})=x_{1}x_{2}$, $g(x_{1},x_{2},x_{3})=x_{1}x_{2}+x_{1}x_{3}+x_{2}x_{3}$. Both functions are permutation invariant. In the rest of this paper, let $V\left(\left\{ x_{j}\right\} _{j\in\mathcal{J}},y\right)$ be a permutation invariant function whose output remains unchanged regardless of the order of the elements $x_{j}(j\in\mathcal{J})$ , but not necessarily invariant to the order of $\left\{ x_{j}\right\} _{j\in\mathcal{J}}$ and $y$. Furthermore, let $\left(x_{1},\cdots,x_{J}\right)$ be a tuple with $J$ elements, where the order of elements matters.

The following definition concerns the continuity of functions with variable-size variables.
\begin{defn}
A function $V:\cup_{J^{\prime}\leq J}\Omega^{J^{\prime}}\subset\mathbb{R}^{I\times(\leq J)}\rightarrow\mathbb{R}$ is continuous if $V|_{\Omega^{J^{\prime}}}$ is continuous for every $J^{\prime}\leq J$.
\end{defn}
Next, we define the function called (multi)symmetric power sum:
\begin{defn}
The function $\eta_{I,J}:[0,\infty)^{I}\rightarrow\mathbb{R}^{\kappa(I,J)}\ \left(\kappa(I,J)\equiv\left(\begin{array}{c}
J+I\\
I
\end{array}\right)-1\right)$ , defined by

\[
\eta_{I,J}^{(s)}(x)=x_{i=1}^{s_{1}}x_{i=2}^{s_{2}}\cdots x_{i=I}^{s_{I}}\ \ \ s\in\left\{ (s_{1},s_{2},\cdots,s_{I})\in\mathbb{Z}_{+}^{I}|1\leq s_{1}+s_{2}+\cdots+s_{I}\leq J\right\} 
\]

is called multisymmetric power sum\footnote{In this definition, we exclude the term such that $s_{i}=0\ \forall i$.} of degree up to $J$ in $I$ variables.
\end{defn}
In the special case $I=1$, this reduces to standard polynomial terms $x,x^{2},\cdots,x^{J}$.

In the following, let $\min\Omega\equiv\left\{ x\in\Omega|x\leq x^{\prime}\ \forall x^{\prime}\in\Omega\right\} $ and $\Omega^{\leq J}\equiv\cup_{J^{\prime}\leq J}\Omega^{J^{\prime}}$. Then, we obtain the following statement:\footnote{It is trivial that functions in the form of $\psi\left(\sum_{k\in\mathcal{J}}\eta_{I,J}(x_{k})\right)$ are permutation invariant. In contrast, whether permutation invariant functions can be represented as $\psi\left(\sum_{k\in\mathcal{J}}\eta_{I,J}(x_{k})\right)$ is not necessarily trivial. }
\begin{prop}
\label{prop:variable_size}Let $\Omega$ be a compact subset of $[0,\infty)^{I}$. Suppose that one of the following conditions holds for a continuous permutation invariant function $V:\Omega^{\leq J}\rightarrow\mathbb{R}$:

(a). $\min\Omega>0_{I}$

(b). $V\left(x_{1},\cdots,x_{J^{\prime}},\underbrace{0,\cdots,0}_{J-J^{\prime}}\right)=V\left(x_{1},\cdots,x_{J^{\prime}}\right)$

Then, there exists a continuous function $\psi:\mathbb{R}^{\kappa(I,J)}\rightarrow\mathbb{R}$ such that:

\[
V\left(\{x_{j}\}_{j\in\mathcal{J}}\right)=\psi\left(\sum_{k\in\mathcal{J}}\eta_{I,J}(x_{k})\right)\ |\mathcal{J}|\leq J.
\]
\end{prop}
Condition (2) intuitively states that additional variables beyond $J^{\prime}$ have no impact on $V$ when they are zero, and they have no effect on the values of ``aggregated variables'' $\sum_{k\in\mathcal{J}}\eta_{I,J}(x_{k})$, and consequently the values of $V$.

Although restricting $\Omega$ to $[0,\infty)^{I}$ might seem limiting, an appropriate change of variables can extend the applicability of this result.\footnote{For instance, for a state $x\in[-M,M]\ (M>0)$, we can define an alternative state $\widetilde{x}\equiv\exp(x)\in[\exp(-M),\exp(M)]$, which always takes a positive value.}
\begin{rem}
\label{rem:wagstaff_2022}\citet{wagstaff2022universal} showed there exist continuous functions $\phi:[0,1]\rightarrow Im(\phi)\subset\mathbb{R}^{J}$ and $\Psi:\left\{ \sum_{k\in\mathcal{J}}\phi(x_{k});x_{k}\in[0,1],|\mathcal{J}|\leq J\right\} \rightarrow\mathbb{R}$ such that $V\left(\{x_{j}\}_{j\in\mathcal{J}}\right)=\Psi\left(\sum_{k\in\mathcal{J}}\phi(x_{k})\right)\ |\mathcal{J}|\leq J$ for a continuous permutation invariant function $V:[0,1]^{\leq J}\rightarrow\mathbb{R}$. Though we would be able to extend the proof to the case with multidimensional variables, the explicit form of $\phi$ is not clear, unlike the case of $\eta$ in Proposition \ref{prop:variable_size}. By using polynomial functions for aggregation, our result provides a more interpretable representation. In addition, by the construction of $\Psi$ and $\phi$, the subset of the domain of $\Psi$, $\left\{ \sum_{k\in\mathcal{J}}\phi(x_{k});x_{k}\in[0,1]\right\} \subset\mathbb{R}^{J}$, should be disjoint for different values of $|\mathcal{J}|$. It implies large domains of $\Psi$, which would be hard to approximate in practice. 

Note that we additionally impose condition (a) or (b), which was not required in \citet{wagstaff2022universal}. However, condition (a) can be easily satisfied by a change of variables and is not particularly restrictive.
\end{rem}
\begin{rem}
We can easily show that $\psi|_{\left\{ \sum_{k\in\mathcal{J}}\eta_{I,J}(x_{k})|x_{k}\in\Omega\right\} }:\left\{ \sum_{k\in\mathcal{J}}\eta_{I,J}(x_{k})|x_{k}\in\Omega\right\} \subset\mathbb{R}^{\kappa(I,J)}\rightarrow\mathbb{R}$ is unique. If $\psi_{1}$ and $\psi_{2}$ satisfy the conditions in Proposition \ref{prop:variable_size} and $V\left(\{x_{j}\}_{j\in\mathcal{J}}\right)=\psi_{1}\left(\sum_{k\in\mathcal{J}}\eta_{I,J}(x_{k})\right)=\psi_{2}\left(\sum_{k\in\mathcal{J}}\eta_{I,J}(x_{k})\right)$ hold for all $x_{k}\in\Omega,k\in\mathcal{J}$, $\psi_{1}$ and $\psi_{2}$ take the same values in the domain $\left\{ \sum_{k\in\mathcal{J}}\eta_{I,J}(x_{k})|x_{k}\in\Omega\right\} $. Hence, the uniqueness of $\psi$ is guaranteed when we restrict the domain to $\left\{ \sum_{k\in\mathcal{J}}\eta_{I,J}(x_{k})|x_{k}\in\Omega\right\} $. Similar discussions also hold for the subsequent statements (Propositions \ref{prop:no_nest_case} and \ref{prop:nest_case}).
\end{rem}
The following simple example shows how Proposition \ref{prop:variable_size} relates to models in economics.
\begin{example}
(Discrete choice model of differentiated products)

Let $\mathcal{J}$ be the set of differentiated products available in a market with $|\mathcal{J}|\leq J$. Define $\log v_{j}$ as the quality of the product produced by firm $j$, where $v_{j}$ is strictly positive for all $j\in\mathcal{J}$. We consider a model where each consumer chooses at most one product that maximizes their utility. Consumer $i$'s utility from purchasing product $j$ is given by $u_{ij}(\epsilon_{ij})\equiv\log v_{j}+\epsilon_{ij}$, while the utility of purchasing nothing is given by $u_{i0}(\epsilon_{i0})\equiv\epsilon_{i0}$. $\epsilon_{ij}$ and $\epsilon_{i0}$ denote idiosyncratic utility shocks. The fraction of consumers purchasing product $j$ is $Pr\left(u_{ij}(\epsilon_{ij})>u_{ik}(\epsilon_{ik})\ \forall k\in\mathcal{J}\cup\{0\}-\{j\}\right)$ where $Pr\left(\cdot\right)$ denotes the probability. If the marginal utility of income is $\alpha$, the expected consumer surplus can be expressed as $W\left(\left\{ v_{j}\right\} _{j\in\mathcal{J}}\right)\equiv\frac{1}{\alpha}E_{\epsilon}\left[\max\left\{ u_{ik}(\epsilon_{ik})\right\} {}_{k\in\mathcal{J}\cup\{0\}}\right]=\frac{1}{\alpha}E_{\epsilon}\left[\max\left\{ \left\{ \log v_{j}+\epsilon_{ij}\right\} _{j\in\mathcal{J}},\epsilon_{i0}\right\} \right]$, where $E_{\epsilon}$ represents the expectation operator over all possible values of $\epsilon$.

By Proposition \ref{prop:variable_size} , if if $W\left(\left\{ v_{j}\right\} _{j\in\mathcal{J}}\right)$ is permutation invariant with respect to $\left\{ v_{j}\right\} _{j\in\mathcal{J}}$, $W\left(\left\{ v_{j}\right\} _{j\in\mathcal{J}}\right)$ can be represented as $W\left(\left\{ v_{j}\right\} _{j\in\mathcal{J}}\right)=\exists\psi\left(\sum_{k\in\mathcal{J}}\eta_{I=1,J}\left(v_{k}\right)\right)=\psi\left(\left\{ \sum_{k\in\mathcal{J}}v_{k}^{m}\right\} _{m=1,\cdots,J}\right)$ . In the setting $\epsilon$ follows a Gumbel distribution, it is known that $W\left(\left\{ v_{j}\right\} _{j\in\mathcal{J}}\right)=\frac{1}{\alpha}\left[\log\left(1+\sum_{k\in\mathcal{J}}\exp\left(\log v_{k}\right)\right)+\gamma\right]=\frac{1}{\alpha}\left[\log\left(1+\sum_{k\in\mathcal{J}}v_{k}\right)+\gamma\right]$ holds, where $\gamma$ denotes the Euler's constant.\footnote{See \citet{train2009discrete} and others.} Thus, when $\epsilon$ follows Gumbel distribution, the expected consumer surplus can be represented using only the first-order unnormalized moment $\sum_{k\in\mathcal{J}}v_{k}$ to represent $W\left(\left\{ v_{j}\right\} _{j\in\mathcal{J}}\right)$ . However, for distributions other than Gumbel, there is no guarantee that a first-order moment is sufficient. Proposition \ref{prop:variable_size} ensures that $W\left(\left\{ v_{j}\right\} _{j\in\mathcal{J}}\right)$ can always be represented using the sum of polynomials up to degree $J$. 
\end{example}
\medskip{}

The next result is more relevant to game-theoretic applications.\footnote{In the context of demand estimation, \citet{singh2025choice} derived a different representation $V\left(x_{j},\left\{ x_{j^{\prime}}\right\} _{j^{\prime}\in\mathcal{J}-\{j\}}\right)=\exists\rho\left(\exists\phi_{1}\left(x_{j}\right)+\sum_{k\in\mathcal{J}-\{j\}}\exists\phi_{2}\left(x_{k}\right)\right)$. However, the domain of $x_{j}$ is restricted to some countable universe, and the explicit forms of $\phi_{1}$ and $\phi_{2}$ are not necessarily clear.}
\begin{prop}
\label{prop:no_nest_case}

Let $\Omega$ be a compact subset of $[0,\infty)^{I}$, and let $\Upsilon$ be a subset of $\mathbb{R}^{C}$. Suppose either of the following conditions holds for a continuous permutation invariant function $V:\Omega\times\Omega^{\leq J-1}\times\Upsilon\rightarrow\mathbb{R}$:

(a). $\min\Omega>0_{I}$

(b). $V\left(x_{j},\left\{ x_{j^{\prime}}\right\} _{j^{\prime}\in\mathcal{J}-\{j\}},y\right)=V\left(x_{j},\left\{ x_{j^{\prime}}\right\} _{j^{\prime}\in\mathcal{J}-\{j\}\ s.t.\ x_{j^{\prime}}\neq0},y\right).$

Then, there exists a continuous function $\psi:\Omega\times\mathbb{R}^{\kappa(I,J-1)}\times\Upsilon\mathbb{\rightarrow\mathbb{R}}$ such that:

\[
V\left(x_{j},\left\{ x_{j^{\prime}}\right\} {}_{j^{\prime}\in\mathcal{J}-\{j\}},y\right)=\psi\left(x_{j},\sum_{k\in\mathcal{J}-\{j\}}\eta_{I,J-1}(x_{k}),y\right)\ |\mathcal{J}|\leq J.
\]
\end{prop}
For the special case, this simplifies to $V\left(x_{j},\left\{ x_{j^{\prime}}\right\} {}_{j^{\prime}\in\mathcal{J}-\{j\}},y\right)=\psi\left(x_{j},\left\{ \sum_{k\in\mathcal{J}-\{j\}}x_{k}^{q}\right\} _{q=1,\cdots,J-1},y\right)\ |\mathcal{J}|\leq J$, because $\eta_{I=1,J-1}(x_{k})=\left\{ x_{k}^{q}\right\} _{q=1,\cdots,J-1}$. In the next section's discussion, we call $\sum_{k\in\mathcal{J}-\{j\}}\eta_{I,J-1}(x_{k})$ (unnormalized) moments based on polynomial functions.

\medskip{}

The next example shows how Proposition \ref{prop:no_nest_case} is related to game-theoretic problems.
\begin{example}
(Model of differentiated products)

Let $\mathcal{J}$ be the set of single-product firms producing differentiated products in a market with $|\mathcal{J}|\leq J$. Difine $p_{j}$ and $\delta_{j}$ as the price and quality of the product produced by firm $j$, respectively. We assume $p_{j}$ and $\delta_{j}$ are both positive. In general, the demand for product $j$ depends on the price and quality of all products in the market. Suppose that the demand for product $j$ is given by $D_{j}\left(p_{j},\delta_{j},\left\{ p_{-j},\delta_{-j}\right\} _{-j\in\mathcal{J}-\{j\}}\right)$, where $\left\{ p_{-j},\delta_{-j}\right\} _{-j\in\mathcal{J}-\{j\}}$ is permutation invariant, meaning that the demand for product $j$ does not depend on the identity of competing products.\footnote{The assumption, which is also called exchangeability, is sometimes utilized in the literature of demand estimation (e.g., \citealp{gandhi2019measuring}, \citealp{compiani2022market}, \citealp{allen2024latent}, \citealp{singh2025choice}).} The profit function of firm $j$ is then expressed as $\pi_{j}\left(p_{j},mc_{j},\delta_{j},\left\{ p_{-j},\delta_{-j}\right\} _{-j\in\mathcal{J}-\{j\}}\right)=\left(p_{j}-mc_{j}\right)D_{j}\left(p_{j},\delta_{j},\left\{ p_{-j},\delta_{-j}\right\} _{-j\in\mathcal{J}-\{j\}}\right)$, where $mc_{j}$ represents the marginal cost of product $j$. Because $\pi_{j}$ is permutation invariant concerning $\left\{ p_{-j},\delta_{-j}\right\} _{-j\in\mathcal{J}-\{j\}}$, Proposition \ref{prop:no_nest_case} implies that $\pi_{j}$ can be represented as $\pi_{j}\left(p_{j},mc_{j},\delta_{j},\left\{ p_{-j},\delta_{-j}\right\} _{-j\in\mathcal{J}-\{j\}}\right)=\exists\psi\left(p_{j},mc_{j},\delta_{j},\sum_{k\in\mathcal{J}-\{j\}}\eta_{I=2,J-1}\left(p_{k},\delta_{k}\right)\right)$. For instance, if $J=3$, $\pi_{j}\left(p_{j},mc_{j},\delta_{j},\left\{ p_{-j},\delta_{-j}\right\} _{-j\in\mathcal{J}-\{j\}}\right)$ can be expressed as $\exists\psi\left(p_{j},mc_{j},\delta_{j},\sum_{k\in\mathcal{J}-\{j\}}p_{k},\sum_{k\in\mathcal{J}-\{j\}}\delta_{k},\sum_{k\in\mathcal{J}-\{j\}}p_{k}^{2},\sum_{k\in\mathcal{J}-\{j\}}\delta_{k}^{2},\sum_{k\in\mathcal{J}-\{j\}}p_{k}\delta_{k}\right)$.
\end{example}

\section{Game-theoretic Applications\label{sec:Economic-Applications}}

This section examines three game-theoretic applications of the mathematical results from the previous section: Two-step policy function estimation, Moment-based Markov equilibrium, and Aggregative games.

\subsection{Two-step policy function estimation\label{subsec:Two-step-policy-function}}

Previous studies have developed practical tools for empirically analyzing oligopolistic markets using a structural approach. One such example is the two-step policy estimation method for static and dynamic incomplete information games (e.g., \citealp{bajari2007estimating}). This method consists of two stages: in the first stage, firms' policy functions---such as those governing entry, exit, or investment decisions---are estimated nonparametrically. In the second stage, structural parameters (e.g., investment cost parameters, entry/exit cost parameters) are estimated using the previously obtained policy functions. Once structural parameters are estimated, various counterfactual scenarios, such as the impact of government policy changes, can be quantitatively evaluated. A key advantage of the two-step method is that it does not require solving for equilibrium to estimate structural parameters, unlike the Nested Fixed-Point (NFXP) approach. Additionally, this method is more robust to the presence of multiple equilibria. Due to its convenience, many empirical studies have adopted this approach, as discussed in \citet{aguirregabiria2021dynamic}.

Consistent estimation of policy functions is crucial for precise estimation and counterfactual simulation in empirical models of games when applying two-step estimation methods. The propositions in the previous section provide insights into the functional form of policy functions used in these estimations.

Consider a dynamic competition model in the style of \citet{Ericson1995}, where firms across all markets follow the same symmetric Markov perfect equilibrium (MPE). In this equilibrium, firms in markets with the same number of competitors share a common policy function $\sigma\left(s_{jm},\left\{ s_{j^{\prime}m}\right\} _{j^{\prime}\neq j},y_{m},\nu_{jm}\right)$, which remains invariant under permutations of competitors' states $\left\{ s_{j^{\prime}m}\right\} _{j^{\prime}\neq j}$. Here, $s_{jm}\in\Omega\subset[0,\infty)^{I}$ represents firm $j$'s states in market $m$, $\nu_{jm}$ denotes the firm's private shock, and $y_{m}$ captures market $m$'s market-level states. Let $\mathcal{J}_{m}$ be the set of firms in market $m$, and let $J$ be the maximum number of firms across all the markets.

Suppose the domain of $s_{jm}$, $\Omega$, and the policy function, satisfy either condition (a) or (b) in Proposition \ref{prop:no_nest_case}. If the firm with $s_{jm}=0$ have a negligible impact on other firms' policy functions, condition (b) holds. Otherwise, a change of variables can be applied to ensure condition (a) is satisfied. Then, by Proposition 2, the policy function can be expressed as $\sigma\left(s_{jm},\left\{ s_{j^{\prime}m}\right\} _{j^{\prime}\neq j},y_{m},\nu_{jm}\right)=\exists\widetilde{\sigma}\left(s_{jm},\sum_{j^{\prime}\in\mathcal{J}_{m}-\{j\}}\eta_{I,J-1}\left(s_{j^{\prime}m}\right),y_{m},\nu_{jm}\right)$ .

It implies probability or density of choosing $a_{jm}$ at state $\left(s_{jm},\left\{ s_{j^{\prime}m}\right\} _{j^{\prime}\neq j},y_{m}\right)$ can be represented as $Pr\left(a_{jm}|s_{jm},\left\{ s_{j^{\prime}m}\right\} _{j^{\prime}\neq j},y_{m}\right)=\exists g\left(s_{jm},\sum_{j^{\prime}\in\mathcal{J}_{m}-\{j\}}\eta_{I,J-1}\left(s_{j^{\prime}m}\right),y_{m}\right)$.\footnote{See \citet{bajari2007estimating} for a discussion on the continuous choice case.} For instance, in the case where $I=1$, the probability can be rewritten as $Pr\left(a_{jm}|s_{jm},\left\{ s_{j^{\prime}m}\right\} _{j^{\prime}\neq j},y_{m}\right)=\exists g\left(s_{jm},\left\{ \sum_{j^{\prime}\in\mathcal{J}_{m}-\{j\}}\left(s_{j^{\prime}m}\right)^{q}\right\} _{q=1,\cdots,J-1},y_{m}\right)$ using a common function $g$ regardless of the number of firms in each market. Hence, we can estimate a policy function as a function of unnormalized moments of competitors' states, even when the number of firms varies across markets.\footnote{Instead we can separately estimate policy functions for markets with the same number of firms. If the number of observations is small and the number of firms varies significantly across markets, this approach may be ineffective.} Because the explicit form of $g$ is unknown, $g$ should be approximated nonparametrically. Though the number of moments should be $\kappa(I,J)$ to exactly represent $Pr\left(a_{jm}|s_{jm},\left\{ s_{j^{\prime}m}\right\} _{j^{\prime}\neq j},y_{m}\right)$, which can be huge for high values of $J$, we can expect adding higher-order moments yields minor differences. In the setting $I=1$, choosing a function $\widetilde{g}\left(s_{jm},\left\{ \sum_{j^{\prime}\in\mathcal{J}_{m}-\{j\}}\left(s_{j^{\prime}m}\right)^{q}\right\} _{q=1,\cdots,K\leq J-1},y_{m}\right)$ might be sufficient to approximate the function $g$ well. Although it depends on empirical contexts, a small number of moments might be enough to approximate the functions well.

Although the use of (unnormalized) moments in policy function estimation has been common in the literature,\footnote{For instance, \citet{Ryan2012} considered a capacity competition model in the cement industry, and estimated firms' investment policy function, as a function of the sum of competitors' capacity.} its formal justification has remained unclear. Without a rigorous foundation, the strategy may fail. The findings in this study formally justify such an approach, provided that the number of moments is sufficiently large and that adding higher-order terms results in only minor differences.\footnote{In the context of nonparametric estimations, using polynomial terms to approximate an unknown continuous function nonparametrically is justified, because there exists an polynomial function sufficiently close to the unknown continuous function by Weierstrass approximation theorem. Analogously, the strategy of using (unnormalized) moments in policy function estimation is justified, because we can exactly represent the function when the number of moments is $\kappa(I,J)$. Note that we might be able to approximate the function well, even when the number of moments is less than $\kappa(I,J)$.}

\subsection{Moment-based Markov equilibrium\label{subsec:Moment-based-Markov-equilibrium}}

\citet{ifrach2017framework} proposed the Moment-based Markov equilibrium (MME) in dynamic oligopoly models with a small number of dominant firms and a large number of fringe firms as an approximation of MPE.\footnote{The concept of MME is related to mean-field games, which is a model where multiple agents interact with each other through an empirical distributions (normalized moments) of other agents' pairs of states and actions (cf. \citealp{lauriere2022learning,han2025deepham}). } One implicit assumption of MPE is that firms monitor all competitiors' state variables. However, as the number of firms increases, the total number of states grows exponentially, making numerical solutions impractical. In contrast, MME reduces the state space by tracking only the states of dominant incumbent firms and a few moments of fringe firms' states, thereby alleviating computational burdens. Due to its efficiency, MME has been applied in recent empirical studies (e.g., \citealp{corbae2021capital,jeon2022learning}).

\citet{ifrach2017framework} showed MME becomes an exact approximation of MPE in the constant returns to scale model. However, the correspondence in more general settings remains unclear. The results in Section \ref{sec:Universal-representation} suggest that symmetric MME is equivalent to MPE when a sufficiently large number of symmetric power sums are used as moments.

To clarify the point, consider a simplified setting without dominant firms or entry/exit decisions.\footnote{Essential ideas would not be lost with this simplification.} Let $x_{j}\in\mathbb{N}^{q}$ be firm $j$'s states, such as its capacity or product quality level. 

Suppose firms follow symmetric MPE. Then, Proposition \ref{prop:no_nest_case} implies that the firms' value function can be reformulated as $V\left(x_{j},\left\{ x_{j^{\prime}}\right\} _{j^{\prime}\in\mathcal{J}-\{j\}}\right)=\overline{V}\left(x_{j},\sum_{j^{\prime}\in\mathcal{J}-\{j\}}\eta_{I,J-1}(x_{j^{\prime}})\right)$. Similarly, firms' investment strategy can be reformulated as a function of $\sum_{j^{\prime}\in\mathcal{J}-\{j\}}\eta_{I,J-1}(x_{j^{\prime}})$. Furthermore, Corollary \ref{cor:chen_et_al} in Appendix \ref{sec:Proof-symmetry} implies there exists an one-to-one mapping between $\left\{ x_{j}\right\} _{j\in\mathcal{J}}$ and $\sum_{j^{\prime}\in\mathcal{J}-\{j\}}\eta_{I,J-1}(x_{j^{\prime}})$. This implies that solving the model using $\sum_{j^{\prime}\in\mathcal{J}-\{j\}}\eta_{I,J-1}(x_{j^{\prime}})$ as state variables is equivalent to solving the MPE---effectively making MME a valid approximation.

Although the degree of moments should be $J-1$ to exactly represent $V\left(x_{j},\left\{ x_{j^{\prime}}\right\} _{j^{\prime}\in\mathcal{J}-\{j\}}\right)$, we can expect adding higher order terms provides smaller information in many cases, as in the discussion of two-step policy function estimation. In such cases, using small number of moments $\left(\sum_{j^{\prime}\in\mathcal{J}-\{j\}}\eta_{I,K\leq J-1}(x_{j^{\prime}});\left\{ \sum_{j^{\prime}\in\mathcal{J}-\{j\}}x_{j^{\prime}}^{q}\right\} _{q=1,\cdots,K\leq J-1}\text{\ if}\ I=1\right)$ is sufficient to approximate the MPE well\footnote{\citet{ifrach2017framework} derived deviation error bounds when using MME as an approximation of MPE.}. 

\subsection{Aggregative games\label{subsec:Aggregative-games}}

The results in Section \ref{sec:Universal-representation} provide insights into aggregative games, which have been widely used to analyze strategic interactions in economic models. Aggregative games are characterized by the fact that each player's payoff depends on their own strategy and an aggregate of all players' strategies. This structure simplifies equilibrium analysis and has been applied in various economic contexts.\footnote{See \citet{jensen2018aggregative} for selective survey of aggregative games.} 

Let $g_{j}:\mathcal{A}\rightarrow\mathbb{R}$ be player $j$'s payoff function, where $\ensuremath{\mathcal{A}}\ensuremath{\equiv}\left\{ \mathcal{A}_{j}\right\} _{j\in\mathcal{J}}$, and $\mathcal{A}_{j}$ denotes player $j$'s strategy set. The following definition formalizes the concept of generalized (fully) aggregative games:
\begin{defn}
(\citealp{cornes2012fully}) The game is called generalized (fully) aggregative, if there exist functions $\widetilde{g_{j}}:\mathcal{A}_{j}\times\mathbb{R}\rightarrow\mathbb{R}$ and $h_{j}:\mathcal{A}_{j}\rightarrow\mathbb{R}$ such that $g_{j}(a)=\widetilde{g_{j}}\left(a_{j},\sum_{j\in\mathcal{J}-\{j\}}h_{j^{\prime}}(a_{j^{\prime}})\right)$.
\end{defn}
We can analogously define:
\begin{defn}
The game is called multidimensional generalized (fully) aggregative, if there exists an integer $K$ and functions $\widetilde{g_{j}}:\mathcal{A}_{j}\times\mathbb{R}^{K}\rightarrow\mathbb{R}$ and $h_{j}:\mathcal{A}_{j}\rightarrow\mathbb{R}^{K}$ such that $g_{j}(a)=\widetilde{g_{j}}\left(a_{j},\sum_{j^{\prime}\in\mathcal{J}-\{j\}}h_{j^{\prime}}(a_{j^{\prime}})\right)$.
\end{defn}
This extends the concept of aggregative games to settings where multiple aggregate variables influence payoffs. Then, we obtain the following statement:
\begin{prop}
\label{prop:aggregative_game}Suppose player $i$'s payoff function $g_{j}$ can be written as $g_{j}(a)=\widetilde{\widetilde{g_{j}}}\left(a_{j},\left\{ a_{j^{\prime}},x_{j^{\prime}}^{0}\right\} _{j^{\prime}\in\mathcal{J}-\{j\}}\right)$, $\widetilde{\widetilde{g_{j}}}$ is continuous, and permutation invariant with respect to $\left\{ a_{j^{\prime}},x_{j^{\prime}}^{0}\right\} _{j^{\prime}\in\mathcal{J}-\{j\}}$. Then, the game is multidimensional generalized aggregative.
\end{prop}
Proposition \ref{prop:aggregative_game} implies all the games satisfying the condition of permutation invariance and continuity of payoff functions are in the form of multidimensional generalized aggregative. Generalized aggregative games, which have been intensively studied in the literature, are the special cases with $K=1$. The case with $K=1$ includes the models with $g_{j}(a)=\widetilde{g_{j}}\left(a_{j},\sum_{j\in\mathcal{J}-\{j\}}a_{j^{\prime}}\right)$. Investigating the setting with $K\geq2$ would enable exploring more general game-theoretic models not studied in the previous studies. I leave it for further research.

\section{Conclusions\label{sec:Conclusions}}

This paper has presented universal representations of symmetric functions with multidimensional variable-size variables, providing a foundation for assessing the validity of approximation methods that aggregate information using moments. It has also explored how these findings offer insights into game-theoretic applications, including two-step policy function estimation, Moment-based Markov Equilibrium (MME), and aggregative games.

While the primary focus has been on these three economic applications, the mathematical results introduced in this study would have broader implications. Future research could uncover additional applications, further extending the relevance of these findings. Investigating how symmetry-based representations can enhance computational efficiency and theoretical modeling in other domains would be a promising direction for further study.

\section*{Acknowledgments}

\selectlanguage{american}%
This paper is based on Chapter 6 of my Ph.D. dissertation at the University of Tokyo. I thank Hiroshi Ohashi, Kei Kawai, and my dissertation committee members for their comments. 

\section*{Declaration of interest statement}

The author has no conflicts of interest to declare that are relevant to the content of this article.

\selectlanguage{english}%
\pagebreak{}

\appendix

\section{Additional results\label{sec:Additional-results-symmetry}}

In this section, we further discuss the implications of the results to other economic models.

\subsection{Merger\label{subsec:Merger}}

Estimated policy functions are sometimes used to simulate counterfactual outcomes. For instance, \citet{benkard2020simulating} employed estimated policy functions to quantitatively compare outcomes with and without a merger in the airline industry allowing for firms' entry/exit decisions, though they did not formally discuss the justifications on the use of a common policy function.\footnote{They estimated a common policy function regardless of the number of firms in each market by using data of many independent markets with varying number of firms.}\footnote{Analogous idea was also applied in \citet{Bruegge2025using}, quantitatively evaluating the price effects of a merger in the airline industry using policy functions on prices.} As discussed below, it is justifiable to assess the effect of a merger by using a common policy function with unnormalized moments.

As discussed in Section \ref{subsec:Two-step-policy-function}, the policy function can be represented as $Pr\left(a_{jm}|s_{jm},\left\{ s_{j^{\prime}m}\right\} _{j^{\prime}\neq j},y_{m}\right)=\exists g\left(s_{jm},\sum_{j^{\prime}\in\mathcal{J}_{m}-\{j\}}\eta_{I,J-1}\left(s_{j^{\prime}m}\right),y_{m}\right)\approx\widetilde{g}\left(s_{jm},\left\{ \sum_{j^{\prime}\in\mathcal{J}_{m}-\{j\}}\left(s_{j^{\prime}m}\right)^{q}\right\} _{q=1,\cdots,K\leq J-1},y_{m}\right)$. Suppose that the function $\widetilde{g}$ is estimated consistently and that a merger leads to an equilibrium under $J_{1}\leq J$ firms to the one under $J_{2}<J_{1}\leq J$ firms. Proposition \ref{prop:no_nest_case} implies that we can simulate the outcomes with and without the merger by using the common policy function $\widetilde{g}$. 

\subsection{Models with dynamic demand\label{subsec:Models-with-dynamic-demand}}

In the models with dynamic demand where the current and future consumer demand is related, firms must track multiple state variables--such as the distribution of heterogeneous consumers' inventory in durable goods--ptimize pricing and investment decisions (e.g., \citealp{goettler2011does}). The results presented in this study also provide insights into such models.

Here, let $B$ be $K$-dimensional state variables firms must monitor. For simplicity, consider the case of a monopolistic firm, and suppose we aim to approximate the firm's value function $V(B)$. When $K$ is large and we don't use any knowledge of the structure of $V$, in general it is difficult to solve the high-dimensional model. Although $V$ may not initially appear to exhibit symmetric properties, symmetry can often be identified based on the model\textquoteright s structure.

Suppose the $k$-th state variable $B^{(k)}$ is parameterized by an $n$-dimensional vector $\theta^{(k)}$, aallowing the value function to be rewritten as $V\left(\left\{ \left(B^{(k)},\theta^{(k)}\right)\right\} _{k=1,\cdots,K}\right)$. Here, $V$ is permutation invariant with respect to $\left(B^{(k)},\theta^{(k)}\right)$. For example, consider a durable goods model where a monopolistic firm tracks discrete consumer types and their product holdings. In this case, state variables correspond to the fraction of each consumer type for each product age, which can be parameterized by preference parameters $\theta_{pref}$ and product age parameters $\theta_{age}$. It is plausible to assume that the order of $\left(B^{(k)},\theta_{pref}^{(k)},\theta_{age}^{(k)}\right)$ does not affect the firm's decision-making.

Then, Proposition \ref{prop:variable_size} implies $V$ can be reformulated as $V\left(\left\{ \left(B^{(k)},\theta^{(k)}\right)\right\} _{k=1,\cdots,K}\right)=\exists\psi\left(\sum_{k=1}^{K}\eta_{1+n,K}\left(B^{(k)},\theta^{(k)}\right)\right)$, where $\psi:\mathbb{R}^{\kappa(1+n,K)}\rightarrow\mathbb{R}$. Hence, we can alternatively use moments $\eta_{1+n,K}\left(B^{(k)},\theta^{(k)}\right)$ as states.

\subsection{Models with Multi-product firms\label{subsec:Models-with-Multi-product}}

Although the results in Section \ref{sec:Universal-representation} do not directly apply to models with multi-product firms, they can be extended to accommodate such settings. The following proposition is particularly relevant for analyzing multi-product firm behavior.\footnote{\citet{nocke2018multiproduct} developed an aggregative games approach to study multiproduct-firm oligopoly. The current study relates to their study, though our results do not assume specific demand structure. } 

\begin{prop}
\label{prop:nest_case}(Nested structure)

Let $\Omega\subset[0,\infty)^{I}$ be a compact subset of $\mathbb{R}^{I}$, and let $\Upsilon$ be a subset of $\mathbb{R}^{C}$. Suppose either of the following conditions holds for a continuous permutation invariant function $V:\Omega^{\leq J}\times\left(\Omega^{\leq J}\right)^{\leq(F-1)}\times\Upsilon\rightarrow\mathbb{R}$:

(a). $\min\Omega>0_{I}$

(b). {\footnotesize{}$V\left(\{x_{j}\}_{j\in\mathcal{J}_{f}},\left\{ \{x_{j}\}_{j\in\mathcal{J}_{\widetilde{f}}}\right\} _{\widetilde{f}\in\mathcal{F-}\{f\}},y\right)=V\left(\{x_{j}\}_{j\in\mathcal{J}_{f}\ s.t.\ x_{j}\neq0},\left\{ \{x_{j}\}_{j\in\mathcal{J}_{\widetilde{f}}\ s.t.\ x_{j}\neq0}\right\} _{\widetilde{f}\in\mathcal{F-}\{f\}\ s.t.\ \neg(x_{j}=0\ \forall j\in\mathcal{J}_{f^{\prime}})},y\right)$}{\footnotesize\par}

Then, there exist continuous functions $\psi_{1}:\mathbb{R}^{\kappa(I,J)}\rightarrow\mathbb{R}^{\kappa(IJ,F-1)}$ and $\psi_{2}:\mathbb{R}^{\kappa(I,J)}\times\mathbb{R}^{\kappa(IJ,F-1)}\times\mathbb{R}^{C}\rightarrow\mathbb{R}$ such that 

\[
V\left(\{x_{j}\}_{j\in\mathcal{J}_{f}},\left\{ \{x_{j}\}_{j\in\mathcal{J}_{\widetilde{f}}}\right\} _{\widetilde{f}\in\mathcal{F-}\{f\}},y\right)=\psi_{2}\left(\sum_{k\in\mathcal{J}_{f}}\eta_{I,J}(x_{k}),\sum_{\widetilde{f}\in\mathcal{F}-\{f\}}\psi_{1}\left(\sum_{k\in\mathcal{J}_{\widetilde{f}}}\eta_{I,J}(x_{k})\right),y\right).
\]
\end{prop}
The proof is shown at the end of this subsection.

Here, consider the dynamic model where each firm decides whether to introduce each product in each period, as in \citet{sweeting2013dynamic}. Let $x_{j}$ be product $j$'s states, namely, whether the product is already introduced at the beginning of the period. We assume that firms follow symmetric MPE, and firm $f$'s value function is permutation invariant with respect to the order of products of the same firms, and to the order of competitors.

The proposition implies that firm $f$'s value function can be reformulated as: 
\[
V\left(\{x_{j}\}_{j\in\mathcal{J}_{f}},\left\{ \{x_{j}\}_{j\in\mathcal{J}_{\widetilde{f}}}\right\} _{\widetilde{f}\in\mathcal{F-}\{f\}}\right)=\exists\psi_{2}\left(\sum_{k\in\mathcal{J}_{f}}\eta_{I,J}(x_{k}),\sum_{\widetilde{f}\in\mathcal{F}-\{f\}}\exists\psi_{1}\left(\sum_{k\in\mathcal{J}_{\widetilde{f}}}\eta_{I,J}(x_{k})\right)\right).
\]
It indicates that value function can be represented as a function of the sum of each rival firm's products' summary statistics $\left(\sum_{\widetilde{f}\in\mathcal{F}-\{f\}}\psi_{1}\left(\sum_{k\in\mathcal{J}_{\widetilde{f}}}\eta_{I,J}(x_{k})\right)\right)$ and the sum of the moments of own firm's products' states $\left(\sum_{k\in\mathcal{J}_{f}}\eta_{I,J}(x_{k})\right)$.

\section{Proof\label{sec:Proof-symmetry}}

\subsection{Proof of Propositions \ref{prop:variable_size} and \ref{prop:no_nest_case}}

Let $\mathcal{M}(I,J;W)\subset M(I,J:\mathbb{R})$ be the set of matrices whose column vector is in $W\subset\mathbb{R}^{I}$, and the rows are sorted based on the descending lexicographical order.\footnote{The sorting guarantees the uniqueness of $\overline{\Lambda}$ in Proposition \ref{prop:chen_et_al}.} Let $a_{j}$ be the $j$-th column vector of the matrix $A\in\mathcal{M}(I,J;W)$, and let $a_{ij}$ be the $(i,j)$-th element of the matrix $A$.

Besides, for a compact set $\Omega\subset[0,\infty)^{I}$, let $\widetilde{\Omega}\subset[0,\infty)^{I}$ be another compact set such that $\Omega\subset\widetilde{\Omega}$, $\min\widetilde{\Omega}=0_{I}$ and $\max\widetilde{\Omega}=\max\Omega$. 
\begin{prop}
\label{prop:chen_et_al}(Based on \citealp{chen2024representation})

Given a compact subset $W\subset\mathbb{R}^{I}$, there exist a continuous function $\overline{\eta}_{I,J}:W\ (\subset\mathbb{R}^{I})\rightarrow Im(\overline{\eta})\ (\subset\mathbb{R}^{K})$ and a unique function $\overline{\Lambda}:\left\{ \sum_{j=1}^{J}\overline{\eta}_{I,J}(a_{j})|a_{j}\in\Omega\right\} \ (\subset\mathbb{R}^{K})\rightarrow Im(\overline{\Lambda})\ (\subset\mathbb{R}^{J})$ such that

\[
\overline{\eta}_{I,J}^{(s)}(A_{j})=a_{i=1}^{s_{1}}a_{i=2}^{s_{2}}\cdots a_{i=I}^{s_{I}}\ \ \ s\in\left\{ (s_{1},s_{2},\cdots,s_{I})\in\mathbb{Z}_{+}|0\leq s_{1}+s_{2}+\cdots+s_{I}\leq J\right\} 
\]

and

\[
A=\overline{\Lambda}\left(\sum_{j=1}^{J}\overline{\eta}_{I,J}(a_{j})\right)\ \forall A\in\mathcal{M}(I,J;W).
\]
\end{prop}
The statement corresponds to Theorem 2.1 of \citet{chen2024representation}, though some notations differ. Using the proposition, we can easily derive the following:
\begin{cor}
\label{cor:chen_et_al}Given a compact subset $W\subset\mathbb{R}^{I}$, there exist a continuous function $\eta_{I,J}:W\ (\subset\mathbb{R}^{I})\rightarrow Im(\eta)\ (\subset\mathbb{R}^{K-1})$ and a unique function $\Lambda:\left\{ \sum_{j=1}^{J}\eta_{I,J}(a_{j})|a_{j}\in\Omega\right\} \ (\subset\mathbb{R}^{K})\rightarrow Im(\Lambda)\ (\subset\mathbb{R}^{J})$ such that

\[
\eta_{I,J}^{(s)}(A_{j})=a_{i=1}^{s_{1}}a_{i=2}^{s_{2}}\cdots a_{i=I}^{s_{I}}\ \ \ s\in\left\{ (s_{1},s_{2},\cdots,s_{I})\in\mathbb{Z}_{+}|1\leq s_{1}+s_{2}+\cdots+s_{I}\leq J\right\} 
\]

,and

\[
A=\Lambda\left(\sum_{j=1}^{J}\eta_{I,J}(a_{j})\right)\ \forall A\in\mathcal{M}(I,J;W).
\]
\end{cor}
\begin{proof}
First, we define a homeomorphism $\gamma:\left\{ \sum_{j=1}^{J}\eta_{I,J}(a_{j})|a_{j}\in\Omega\right\} \rightarrow\left\{ \sum_{j=1}^{J}\overline{\eta}_{I,J}(a_{j})|a_{j}\in\Omega\right\} $ such that $\gamma(x_{1},\cdots,x_{K-1})=(J,x_{1},\cdots,x_{K-1})$. Then, $\sum_{j=1}^{J}\overline{\eta}_{I,J}(a_{j})=\gamma\left(\sum_{j=1}^{J}\eta_{I,J}(a_{j})\right)$ holds because $\sum_{j=1}^{J}\overline{\eta}_{I,J}^{\left(s=(0,\cdots,0)\right)}(x_{j})=\sum_{j=1}^{J}1=J$, and we have $A=\overline{\Lambda}\left(\sum_{j=1}^{J}\overline{\eta_{I,J}}(a_{j})\right)=\overline{\Lambda}\left(\gamma^{-1}\left(\sum_{j=1}^{J}\eta_{I,J}(a_{j})\right)\right)$. Hence, by defining a function $\Lambda\equiv\overline{\Lambda}\circ\gamma^{-1}$, we obtain the statement.
\end{proof}
The next lemma is used to extend the results above to the case of functions with variable-size variables.
\begin{lem}
\label{lem:function_extension}For a compact subset $\Omega\subset[0,\infty)^{I}$ such that $\min\Omega>0$, let $\widehat{\Omega}\equiv\Omega\cup\prod_{i=1}^{I}\left[0,\min\Omega_{i}\right]$. For a continuous permutation invariant function $V:\Omega^{\leq J}\rightarrow\mathbb{R}$, define a function $\overline{V}:\mathcal{M}(I,J,\widehat{\Omega})\rightarrow\mathbb{R}$ such that $\overline{V}(x_{1},\cdots,x_{J})=\sum_{\left\{ j\in\{1,\cdots,J\}:x_{j}\in\Omega\right\} \subset\widetilde{\mathcal{J}}\subset\{1,\cdots,J\}}\left[\prod_{j\in\widetilde{\mathcal{J}}}a(x_{j})\right]\cdot\left[\prod_{j\in\{1,\cdots,J\}-\widetilde{\mathcal{J}}}b(x_{j})\right]\cdot V\left(\left\{ \max\{x_{j},\min\Omega\}\right\} {}_{j\in\widetilde{\mathcal{J}}}\right)$,

where $a(x)\equiv\prod_{i}\min\left\{ 1,\frac{x^{(i)}}{\min\Omega^{(i)}}\right\} \ (x\in[0,\infty)^{I})$ and $b(x)\equiv\prod_{i}\max\left\{ 0,\frac{\min\Omega^{(i)}-x^{(i)}}{\min\Omega^{(i)}}\right\} \ (x\in[0,\infty)^{I})$.

Then, 

(a).$\overline{V}$ is continuous, and

(b).$\overline{V}$ satisfies $\overline{V}\left(x_{1},\cdots,x_{J^{\prime}},\underbrace{0_{I},\cdots,0_{I}}_{J-J^{\prime}}\right)=V\left(\left\{ x_{j}\right\} _{j=1,\cdots,J^{\prime}}\right)\ \ \ (x_{j}\in\Omega;\ j=1,\cdots,J^{\prime}).$
\end{lem}
\begin{proof}
\textbf{Proof of (a).}

It suffices to show $\lim_{x_{J}\uparrow m}\overline{V}(x_{1},\cdots,x_{J-1},x_{J})=\overline{V}(x_{1},\cdots,x_{J-1},x_{J}=\min\Omega)$.

Because $\lim_{x_{J}\uparrow m}b(x_{J})=b(m)=0$, the terms associated with $\widetilde{\mathcal{J}}$ such that $J\notin\widetilde{\mathcal{J}}$ disappear when taking the limit:

\begin{eqnarray*}
 &  & \lim_{x_{J}\uparrow m}\overline{V}(x_{1},\cdots,x_{J-1},x_{J}<m)\\
 & = & \lim_{x_{J}\uparrow m}\sum_{\left\{ j\in\{1,\cdots,J-1\}:x_{j}\in\Omega\right\} \subset\widetilde{\mathcal{J}}\subset\{1,\cdots,J\}}\left[\prod_{j\in\widetilde{\mathcal{J}}}a(x_{j})\right]\cdot\left[\prod_{j\in\{1,\cdots,J\}-\widetilde{\mathcal{J}}}b(x_{j})\right]\cdot V\left(\left\{ \max\{x_{j},\min\Omega\}\right\} _{j\in\mathcal{J}}\right)\\
 & = & \sum_{\left\{ j\in\{1,\cdots,J-1\}:x_{j}\in\Omega\right\} \cup\{J\}\subset\widetilde{\mathcal{J}}\subset\{1,\cdots,J\}}\left[\prod_{j\in\widetilde{\mathcal{J}}}a(x_{j})\right]\cdot\left[\prod_{j\in\{1,\cdots,J\}-\widetilde{\mathcal{J}}}b(x_{j})\right]\cdot V\left(\left\{ \max\{x_{j},\min\Omega\}\right\} _{j\in\mathcal{J}}\right)\\
 & = & \overline{V}(x_{1},\cdots,x_{J-1},x_{J}=\min\Omega).
\end{eqnarray*}

Hence, we obtain the statement.

\textbf{Proof of (b).}

Let $x_{1},\cdots,x_{J^{\prime}}\in\Omega$. Because $a(0_{I})=0,$ the terms associated with $\widetilde{\mathcal{J}}\neq\left\{ j\in\{1,\cdots,J\}:x_{j}\in\Omega\right\} $, i.e. $\widetilde{\mathcal{J}}\neq\{1,\cdots,J^{\prime}\}$ disappear, and

\begin{eqnarray*}
 &  & \overline{V}(x_{1},\cdots,x_{J^{\prime}},0_{I}\cdots,0_{I})\\
 & = & \left[\prod_{j\in\{1,\cdots,J^{\prime}\}}a(x_{j})\right]\cdot\left[\prod_{j\in\{1,\cdots,J\}-\{1,\cdots,J^{\prime}\}}b(x_{j})\right]\cdot V\left(\left\{ \max\{x_{j},\min\Omega\}\right\} _{j\in\mathcal{J}}\right)\\
 & = & \left[\prod_{j\in\mathcal{J}}1\right]\cdot\left[\prod_{j\in\{1,\cdots,J\}-\{1,\cdots,J^{\prime}\}}1\right]\cdot V\left(\left\{ x_{j}\right\} _{j=1,\cdots,J^{\prime}}\right)\ \ \ \left(\because\ a(x_{j})=1\ \text{for}\ x_{j}\in\Omega,\ b(0_{I})=1\right)\\
 & = & V\left(\left\{ x_{j}\right\} _{j=1,\cdots,J^{\prime}}\right).
\end{eqnarray*}
\end{proof}

\subsubsection{Proof of Proposition \ref{prop:variable_size}}
\begin{proof}
First, by Tierze extension theorem, we can take a function such that $V:\cup_{J^{\prime}=1}^{J}\left[\prod_{i=1}^{I}\left[\min\Omega_{i},\max\Omega_{i}\right]\right]^{J^{\prime}}\rightarrow\mathbb{R}$. If $\min\Omega>0$, by Lemma \ref{lem:function_extension}, we can construct a continuous function $\overline{V}:\mathcal{M}(I,J;\widetilde{\Omega})\ (\subset\mathbb{R}^{I\times J})\rightarrow\mathbb{R}$ such that $\overline{V}\left(x_{1},\cdots,x_{J^{\prime}},\underbrace{0_{I},\cdots,0_{I}}_{J-J^{\prime}}\right)=V\left(\left\{ x_{j}\right\} _{j=1,\cdots,J^{\prime}}\right)$, for the function $V:\Omega^{\leq J}\rightarrow\mathbb{R}$. Let $\overline{V}=V$ if condition (b) holds.

By Corollary \ref{cor:chen_et_al}, there exists a continuous function and $\widetilde{\Lambda}:\left\{ \sum_{j=1}^{J}\eta_{I,J}(x_{j})|x_{j}\in\widetilde{\Omega}\subset\mathbb{R}^{I}\right\} \rightarrow Im(\Lambda)\ (\subset\mathbb{R}^{J})$ such that $\left\{ x_{1},\cdots,x_{J^{\prime}},\underbrace{0_{I},\cdots,0_{I}}_{J-J^{\prime}}\right\} =\widetilde{\Lambda}\left(\sum_{j=1}^{J^{\prime}}\eta_{I,J}(x_{j})\right)$. Hence, by defining $\psi\equiv\overline{V}\circ\widetilde{\Lambda}$, we obtain:

\begin{eqnarray*}
V\left(\left\{ x_{1},\cdots,x_{J^{\prime}}\right\} \right) & = & \overline{V}\left(x_{1},\cdots,x_{J^{\prime}},\underbrace{0_{I},\cdots,0_{I}}_{J-J^{\prime}}\right)\\
 & = & \overline{V}\left(\widetilde{\Lambda}\left(\sum_{j=1}^{J^{\prime}}\eta_{I,J}(x_{j})\right)\right)=\psi\left(\sum_{j=1}^{J^{\prime}}\eta_{I,J}(x_{j})\right).
\end{eqnarray*}
\end{proof}

\subsubsection{Proof of Proposition \ref{prop:no_nest_case}}
\begin{proof}
As in the proof of Proposition \ref{prop:variable_size}, we can construct a continuous function $\overline{V}:\Omega\times\mathcal{M}(I,J-1;\widetilde{\Omega})\times\Upsilon\rightarrow\mathbb{R}$ such that $\overline{V}\left(x_{j},\left\{ x_{k}\right\} _{k\in\mathcal{J}-\{j\}},y\right)=V\left(x_{j},\left\{ x_{k}\right\} _{k\in\mathcal{J}-\{j\}\ s.t.\ x_{k}\neq0},y\right)$, for the function $V:\Omega\times\Omega^{\leq J-1}\times\Upsilon\rightarrow\mathbb{R}$. Hence, 

\begin{eqnarray*}
V\left(x_{j},\left\{ x_{j}\right\} _{k\in\mathcal{J}-\{j\}},y\right) & = & \overline{V}\left(x_{j},\left(\left\{ x_{k}\right\} _{k\in\mathcal{J}-\{j\}},\underbrace{0,\cdots,0}_{J-|\mathcal{J}|}\right),y\right)\\
 & = & \overline{V}\left(x_{j},\exists\widetilde{\Lambda}\left(\sum_{k\in\mathcal{J}-\{j\}}\eta_{I,J-1}(x_{k})\right),y\right)\ \left(\because\text{Corollary }\ref{cor:chen_et_al}\right)\\
 & = & \exists\psi\left(x_{j},\sum_{k\in\mathcal{J}-\{j\}}\eta_{I,J-1}(x_{k}),y\right).
\end{eqnarray*}
\end{proof}

\subsection{Proof of Proposition \ref{prop:aggregative_game}}
\begin{proof}
By Proposition \ref{prop:no_nest_case}, we can take functions $\widetilde{g_{j}}:\mathcal{A}_{j}\times\mathbb{R}^{K}\rightarrow\mathbb{R}$ and $\widetilde{h_{j}}:\mathcal{A}_{j}\times\{x_{j}^{0}\}\rightarrow\mathbb{R}^{K}$ such that $\widetilde{\widetilde{g_{j}}}\left(a_{j},\left\{ a_{j^{\prime}},x_{j^{\prime}}^{0}\right\} _{j^{\prime}\in\mathcal{J}-\{j\}}\right)=\widetilde{g_{j}}\left(a_{j},\sum_{j^{\prime}\in\mathcal{J}-\{j\}}\widetilde{h_{j^{\prime}}}(a_{j^{\prime}},x_{j^{\prime}}^{0})\right)$. By defining $h_{j^{\prime}}:\mathcal{A}_{j}\rightarrow\mathbb{R}^{K}$ such that $h_{j^{\prime}}\equiv\widetilde{h_{j^{\prime}}}\left(a_{j^{\prime}},x_{j^{\prime}}^{0}\right)$, we obtain the statement.
\end{proof}

\subsection{Proof of Proposition \ref{prop:nest_case}}
\begin{proof}
First, for a $I\times J\times F$ dimensional array $A,$ let $A_{f}\equiv A[:,:,f]\subset\mathbb{R}^{I\times J}$ and $a_{jf}\equiv A[:,j,f]$. We further define $A_{-f}\subset\mathbb{R}^{I\times J\times(F-1)}$ which corresponds to $A$ skipping $A_{f}$. Analogous to the case of two dimensional matrices, let $\mathcal{M}(I,J,F;W)\subset M(I,J,F:\mathbb{R})$ be the set of $I\times J\times F$ dimensional arrays whose column vectors $(A[:,j,f]\ j=1,\cdots,J,\ f=1,\cdots,F)$ are in $W\subset\mathbb{R}^{I}$, and they are sorted based on the descending lexicographical order.

Then, by Corollary \ref{cor:chen_et_al}, there exists a unique homeomorphism $\Psi_{1}:\left\{ \sum_{j=1}^{J}\eta_{I,J}(a_{jf})|a_{jf}\in\widetilde{\Omega}\subset\mathbb{R}^{I}\right\} \rightarrow Im(\Psi_{1})\ (\subset\mathbb{R}^{I\times J})$ such that:

\begin{equation}
A_{f}=\Psi_{1}\left(\sum_{j=1}^{J}\eta_{I,J}(a_{jf})\right)\ \forall f,A_{f}\in\mathcal{M}(I,J;\widetilde{\Omega}).\label{eq:A_f-1}
\end{equation}
Furthermore, there exists a unique homeomorphism $\Psi_{2}:\left\{ \sum_{f^{\prime}\in\{1,\cdots,F\}-\{f\}}\eta_{IJ,F-1}(A_{f^{\prime}})|A_{f^{\prime}}\in\widetilde{\Omega}^{J}\subset\mathbb{R}^{I\times J}\right\} \rightarrow Im(\Psi_{2})\ (\subset\mathbb{R}^{I\times J\times(F-1)})$ such that:

\begin{equation}
A_{-f}=\Psi_{2}\left(\sum_{f^{\prime}\in\{1,\cdots,F\}-\{f\}}\eta_{IJ,F-1}(A_{f^{\prime}})\right)\ \forall A_{-f}\in\mathcal{M}(I,J,F-1;\widetilde{\Omega}).\label{eq:A_minus_f-1}
\end{equation}

As in the case of Proposition \ref{prop:variable_size}, we can construct a continuous permutation invariant function $\overline{V}:\mathcal{M}(I,J;\widetilde{\Omega})\times\mathcal{M}(I,J,F-1;\widetilde{\Omega})\times\Upsilon\rightarrow\mathbb{R}$ such that $\overline{V}\left(A_{f},A_{-f},y\right)=V\left(\{x_{j}\}_{j\in\mathcal{J}_{f}},\left\{ \{x_{j}\}_{j\in\mathcal{J}_{\widetilde{f}}}\right\} _{\widetilde{f}\in\mathcal{F-}\{f\}},y\right)$ for the function $V:\Omega^{\leq J}\times\left(\Omega^{\leq J}\right)^{\leq(F-1)}\times\Upsilon\rightarrow\mathbb{R}$. Here, $A_{f}$ is a matrix where $\left\{ x_{j\in\mathcal{J}_{f}},\underbrace{0_{I},\cdots,0_{I}}_{J-|\mathcal{J}|}\right\} $ are sorted based on the descending lexicographical order, and $A_{-f}$ is a matrix where $\left\{ A_{\widetilde{f}\in\mathcal{F}-\{f\}},\underbrace{0_{I\times J},\cdots,0_{I\times J}}_{F-|\mathcal{F}|}\right\} $ are sorted based on the descending lexicographical order.

Using equations (\ref{eq:A_f-1}) and (\ref{eq:A_minus_f-1}), we obtain:

\begin{eqnarray*}
 &  & \overline{V}\left(A_{f},A_{-f},y\right)\\
 & = & \overline{V}\left(\Psi_{1}\left(\sum_{j=1}^{J}\eta_{I,J}(a_{jf})\right),\Psi_{2}\left(\sum_{f^{\prime}\in\{1,\cdots,F\}-\{f\}}\eta_{IJ,F-1}\left(\Psi_{1}\left(\sum_{j=1}^{J}\eta_{I,J}(a_{jf^{\prime}})\right)\right)\right),y\right)\\
 & = & \exists\psi_{2}\left(\sum_{j=1}^{J}\eta_{I,J}(a_{jf}),\sum_{f^{\prime}\in\{1,\cdots,F\}-\{f\}}\exists\psi_{1}\left(\sum_{j=1}^{J}\eta_{I,J}(a_{jf^{\prime}})\right),y\right).
\end{eqnarray*}

Hence, by the relation between $V$ and $\overline{V}$, we obtain the statement.
\end{proof}
\bibliographystyle{apalike}
\bibliography{literature_symmetry}

\begin{thebibliography}{}

\bibitem[Aguirregabiria et~al., 2021]{aguirregabiria2021dynamic}
Aguirregabiria, V., Collard-Wexler, A., and Ryan, S.~P. (2021).
\newblock Dynamic games in empirical industrial organization.
\newblock In {\em Handbook of industrial organization}, volume~4, pages 225--343. Elsevier.

\bibitem[Allen and Rehbeck, 2024]{allen2024latent}
Allen, R. and Rehbeck, J. (2024).
\newblock Latent utility and permutation invariance: A revealed preference approach.
\newblock {\em Journal of Econometrics}, 244(1):105844.

\bibitem[Bajari et~al., 2007]{bajari2007estimating}
Bajari, P., Benkard, C.~L., and Levin, J. (2007).
\newblock Estimating dynamic models of imperfect competition.
\newblock {\em Econometrica}, 75(5):1331--1370.

\bibitem[Benkard et~al., 2020]{benkard2020simulating}
Benkard, C.~L., Bodoh-Creed, A., and Lazarev, J. (2020).
\newblock Simulating the dynamic effects of horizontal mergers: Us airlines.

\bibitem[Bruegge et~al., 2025]{Bruegge2025using}
Bruegge, C., Gowrisanaran, G., and Gross, A. (2025).
\newblock {Using Policy Functions to Estimate Merger Impacts: an Application to JetBlue-Spirit}.

\bibitem[Chen et~al., 2024]{chen2024representation}
Chen, C., Chen, Z., and Lu, J. (2024).
\newblock Representation theorem for multivariable totally symmetric functions.
\newblock {\em Communications in Mathematical Sciences}, 22(5):1195--1201.

\bibitem[Compiani, 2022]{compiani2022market}
Compiani, G. (2022).
\newblock Market counterfactuals and the specification of multiproduct demand: A nonparametric approach.
\newblock {\em Quantitative Economics}, 13(2):545--591.

\bibitem[Corbae and D'Erasmo, 2021]{corbae2021capital}
Corbae, D. and D'Erasmo, P. (2021).
\newblock Capital buffers in a quantitative model of banking industry dynamics.
\newblock {\em Econometrica}, 89(6):2975--3023.

\bibitem[Cornes and Hartley, 2012]{cornes2012fully}
Cornes, R. and Hartley, R. (2012).
\newblock Fully aggregative games.
\newblock {\em Economics Letters}, 116(3):631--633.

\bibitem[Doraszelski and Pakes, 2007]{doraszelski2007framework}
Doraszelski, U. and Pakes, A. (2007).
\newblock A framework for applied dynamic analysis in {IO}.
\newblock {\em Handbook of industrial organization}, 3:1887--1966.

\bibitem[Ericson and Pakes, 1995]{Ericson1995}
Ericson, R. and Pakes, A. (1995).
\newblock {Markov-perfect industry dynamics: A framework for empirical work}.
\newblock {\em Review of Economic Studies}, 62(1):53--82.

\bibitem[Gandhi and Houde, 2019]{gandhi2019measuring}
Gandhi, A. and Houde, J.-F. (2019).
\newblock Measuring substitution patterns in differentiated-products industries.
\newblock {\em NBER Working Paper}, (w26375).

\bibitem[Goettler and Gordon, 2011]{goettler2011does}
Goettler, R.~L. and Gordon, B.~R. (2011).
\newblock {Does AMD spur Intel to innovate more?}
\newblock {\em Journal of Political Economy}, 119(6):1141--1200.

\bibitem[Han et~al., 2022]{han2022universal}
Han, J., Li, Y., Lin, L., Lu, J., Zhang, J., and Zhang, L. (2022).
\newblock Universal approximation of symmetric and anti-symmetric functions.
\newblock {\em Communications in Mathematical Sciences}, 20:1397--1408.

\bibitem[Han and Yang, 2025]{han2025deepham}
Han, J. and Yang, Yucheng~Weinan, E. (2025).
\newblock {DeepHAM: A Global Solution Method for Heterogeneous Agent Models With Aggregate Shocks}.

\bibitem[Hotz and Miller, 1993]{hotz1993conditional}
Hotz, V.~J. and Miller, R.~A. (1993).
\newblock Conditional choice probabilities and the estimation of dynamic models.
\newblock {\em The Review of Economic Studies}, 60(3):497--529.

\bibitem[Ifrach and Weintraub, 2017]{ifrach2017framework}
Ifrach, B. and Weintraub, G.~Y. (2017).
\newblock A framework for dynamic oligopoly in concentrated industries.
\newblock {\em The Review of Economic Studies}, 84(3):1106--1150.

\bibitem[Jensen, 2018]{jensen2018aggregative}
Jensen, M.~K. (2018).
\newblock Aggregative games.
\newblock In {\em Handbook of Game Theory and Industrial Organization, Volume I}, pages 66--92. Edward Elgar Publishing.

\bibitem[Jeon, 2022]{jeon2022learning}
Jeon, J. (2022).
\newblock Learning and investment under demand uncertainty in container shipping.
\newblock {\em The RAND Journal of Economics}, 53(1):226--259.

\bibitem[Kahou et~al., 2021]{kahou2021exploiting}
Kahou, M.~E., Fern{\'a}ndez-Villaverde, J., Perla, J., and Sood, A. (2021).
\newblock Exploiting symmetry in high-dimensional dynamic programming.
\newblock Technical report, National Bureau of Economic Research.

\bibitem[Lauri{\`e}re et~al., 2022]{lauriere2022learning}
Lauri{\`e}re, M., Perrin, S., P{\'e}rolat, J., Girgin, S., Muller, P., {\'E}lie, R., Geist, M., and Pietquin, O. (2022).
\newblock Learning in mean field games: A survey.
\newblock {\em arXiv preprint arXiv:2205.12944}.

\bibitem[Nocke and Schutz, 2018]{nocke2018multiproduct}
Nocke, V. and Schutz, N. (2018).
\newblock Multiproduct-firm oligopoly: An aggregative games approach.
\newblock {\em Econometrica}, 86(2):523--557.

\bibitem[Pakes and McGuire, 1994]{Pakes_McGuire1994}
Pakes, A. and McGuire, P. (1994).
\newblock {Computing Markov-Perfect Nash Equilibria: Numerical Implications of a Dynamic Differentiated Product Model}.
\newblock {\em The RAND Journal of Economics}, 25(4):555--589.

\bibitem[Ryan, 2012]{Ryan2012}
Ryan, S.~P. (2012).
\newblock {The Costs of Environmental Regulation in a Concentrated Industry}.
\newblock {\em Econometrica}, 80(3):1019--1061.

\bibitem[Singh et~al., 2025]{singh2025choice}
Singh, A., Liu, Y., and Yoganarasimhan, H. (2025).
\newblock Choice models and permutation invariance: Demand estimation in differentiated products markets.
\newblock {\em mimeo}.

\bibitem[Sweeting, 2013]{sweeting2013dynamic}
Sweeting, A. (2013).
\newblock Dynamic product positioning in differentiated product markets: The effect of fees for musical performance rights on the commercial radio industry.
\newblock {\em Econometrica}, 81(5):1763--1803.

\bibitem[Train, 2009]{train2009discrete}
Train, K.~E. (2009).
\newblock {\em Discrete choice methods with simulation}.
\newblock Cambridge university press.

\bibitem[Wagstaff et~al., 2022]{wagstaff2022universal}
Wagstaff, E., Fuchs, F.~B., Engelcke, M., Osborne, M.~A., and Posner, I. (2022).
\newblock Universal approximation of functions on sets.
\newblock {\em Journal of Machine Learning Research}, 23(151):1--56.

\bibitem[Zaheer et~al., 2017]{zaheer2017deep}
Zaheer, M., Kottur, S., Ravanbakhsh, S., Poczos, B., Salakhutdinov, R.~R., and Smola, A.~J. (2017).
\newblock Deep sets.
\newblock {\em Advances in neural information processing systems}, 30.

\end{thebibliography}

\end{document}